\documentclass[letterpaper,11pt]{article}

\usepackage[runin]{abstract}
\usepackage{geometry}
\usepackage{titlesec}
\usepackage{graphicx}
\usepackage{amsmath}
\usepackage{amsthm}
\usepackage{amsfonts}
\usepackage{amssymb}
\usepackage{url}

\usepackage[font=sf]{caption}

\titleformat*{\section}{\Large\bfseries}
\titleformat*{\subsection}{\large\sc}
\titleformat*{\subsubsection}{\itshape}

\begin{document}

\title{{\bf Jump-starting coordination in a stag hunt:\\ Motivation, mechanisms, and their analysis}}

\author{{\large{ Ioannis Avramopoulos}}}

\date{}

\maketitle

\thispagestyle{empty} 

\newtheorem{definition}{Definition}
\newtheorem{proposition}{Proposition}
\newtheorem{theorem}{Theorem}
\newtheorem{corollary}{Corollary}
\newtheorem{lemma}{Lemma}
\newtheorem{axiom}{Axiom}
\newtheorem{thesis}{Thesis}

\vspace*{-0.2truecm}

\begin{abstract}
The stag hunt (or assurance game) is a simple game that has been used as a prototype of a variety of social coordination problems (ranging from the social contract to the adoption of technical standards). Players have the option to either use a superior cooperative strategy whose payoff depends on the other players' choices or use an inferior strategy whose payoff is independent of what other players do; the cooperative strategy may incur a loss if sufficiently many other players do not cooperate. Stag hunts have two (strict) pure Nash equilibria, namely, universal cooperation and universal defection (as well as a mixed equilibrium of low predictive value). Selection of the inferior (pure) equilibrium is called a {\em coordination failure}. In this paper, we present and analyze using game-theoretic techniques mechanisms aiming to avert coordination failures and incite instead selection of the superior equilibrium. Our analysis is based on the solution concepts of Nash equilibrium, dominance solvability, as well as a formalization of the notion of ``incremental deployability,'' which is shown to be keenly relevant to the sink equilibrium.
\end{abstract}

\section{Introduction}

The motivation for this paper is a phenomenon that manifests at the core architecture of the Internet, namely, that technologies (e.g., network protocols) emerge promising benefits for everyone over the status quo, however, adoption fails. Consider, for example, Internet routing, i.e., the network function responsible for delivering data (packets) between Internet hosts along network paths. That BGP (the routing protocol providing global connectivity) is vulnerable to accidental misconfiguration and malicious attack is not a particularly novel observation, but that Internet routing remains vulnerable despite significant efforts by the research community, governmental organizations, and service providers to effect change in BGP is arguably worrisome.

A recurring theme in the enterprises to effect fundamental change in the Internet is a frequently encountered dependence of their probability of success on {\em positive externalities} in the following sense: Unilateral action (by a single individual or organization) is, generally, harmful to the corresponding individual or organization, and multilateral effort is necessary for the venture to be effective. Consider, for example, the problem of deploying a secure version of BGP such as S-BGP \cite{Kent}. Unilateral S-BGP deployment is costly to the corresponding organization (typically an ``autonomous system'') without providing security benefits, however, as the autonomous systems that deploy S-BGP increase security benefits quickly ramp up to outrun those of BGP.

This paper rests on the thesis that the adoption of emerging technologies in a social system is a matter that transcends technical challenges and critically relies on the {\em incentives} that shape the adoption environment (as measured by the utilities that adoption decisions confer to {\em individual} members of the social system, rather than, for example, on aggregate social utility metrics).

We note that we are not the first to advocate such a thesis: Long prior to the Internet's innovation stalemate, there was a general (pessimistic) belief that the question of {\em managing common-pool resources,} (that is, resource allocation problems modeled after patterns akin to the prisoner's dilemma) is inherently defiant of admitting technical, if any, solutions \cite{Hardin}. (This pessimistic belief was later refuted, in part, by Ostrom's empirical work on governing structures intended to solve collective action problems \cite{Ostrom}.) In a more recent manifestation of this thesis, we believe that, by and large, the networking community's effort to render Internet innovation {\em incrementally deployable} (a term that will play an instrumental role in our analysis in the sequel) amounts to acceptance of a homologous, if not identical, premise. In the same vein, Chan {\em et al.} \cite{Adoptability} propose an {\em adoptability model} of technological innovation (of Internet-based technologies such as routing protocols) that factors player utilities, and leverage their model to experiment with protocol design.

\subsection{Our thesis}

The eventide of the 20th and the advent of the 21st century was signaled by the emergence of socially beneficial peer-produced technologies such as the Web, Linux, Wikipedia, and Twitter whose success depends on positive externalities. The proliferation of these technologies became possible owing to the widespread penetration of the Internet, which, in more general terms, has incited the manifestation of techno-economic phenomena of an unprecedented nature (cf. {\em social production} \cite{Benkler}) challenging traditional market-based production paradigms as evidenced by Wikipedia's or open-source software's success. Similarly, the production of the Internet as a global communications service is, arguably, itself a techno-economic phenomenon of an unprecedented nature; the Internet is not a technology in the common understanding of this term, but rather an artifact where {\em social phenomena} (akin to those that manifest in social production) are critical to its operation. 


The following question emerges naturally: How is it possible to explain Wikipedia's success vis-\'a-vis the inertia in deploying countermeasures against routing attacks, for example? This question does not admit a simple answer, but the following model sheds light on the problem at hand.

The {\em stag hunt} (also sometimes referred to in the literature as {\em assurance game}) is a simple game that has been used as a model in a variety of social coordination problems most prominently among them it has been used as a prototype of the social contract. The model appears in Rousseau's {\em A Discourse on Inequality}. In a stag hunt, each player must choose one of two strategies, namely, to either {\em cooperate} in choosing a superior strategy whose benefit, however, critically depends on whether other players choose alike or instead {\em defect} and choose a less attractive status quo strategy that reaps moderate benefits independently of what other players do; selecting the superior strategy may incur a loss if sufficiently many other players do not cooperate.

Like all coordination games, stag hunts have multiple equilibria; in particular, universal cooperation in using the (Pareto) superior strategy and universal defection from using that strategy (using instead the inferior status quo strategy) are the only pure-strategy Nash equilibria in this game.\footnote{There is also a mixed equilibrium that provides payoffs inferior than either of the pure equilibria, is deemed unstable, and is generally ignored in the literature.} Given the payoff asymmetry of the superior and inferior equilibrium, which equilibrium is selected is a question of evident interest in the analysis of the game  that has been extensively studied both theoretically and empirically; predicting the outcome of a stag hunt is a notoriously hard {\em equilibrium selection problem} (cf. \cite{HS}). 

The stag hunt aptly captures the incentive structure in environments where superior technologies whose success depends on positive externalities compete with (inferior) incumbent (status quo) technologies. However, equilibrium selection theory is not general in a position to provide accurate predictions on the outcome of technological competition whose incentive structure falls into the aforementioned pattern. Furthermore, even if game theory reaches a mature enough state wherein such predictions become possible, designers may not, in general, find themselves in the position to rely on evolutionary forces alone to supplant parochial architectural elements. But that is not to say we are left without recourse in fostering the proliferation of superior technologies.

I believe computer scientists have overlooked an important piece of the puzzle regarding effecting architectural innovation in the Internet, namely, that the Internet is a {\em population} of independent self-interested actors who need to {\em coordinate} their activities to effect change. In this vein, this paper proposes a theory of architectural evolution which views technologies in the abstract as black boxes and posits that the failure to effect innovation in the (generally acknowledge as parochial) basic architecture of the Internet amounts to {\em coordination failures,} that is, equilibrium selection phenomena in coordination games, which can be reasonably approximated by stag hunts.

\subsection{Diffusing technological innovation}

This abstract perspective is rather common in models of the diffusion of innovation. Quoting Arthur \cite{Technology} ``The people who have thought hardest about the general questions of technology have mostly been social scientists and philosophers, and understandably they have tended to view technology from the outside as stand-alone objects. \ldots \hspace{0.2 mm} Seeing technology this way, from the outside, works well enough if we want to know how technologies enter the economy and spread within it.'' I believe that the bottleneck in Internet innovation is not a lack of creativity in devising new technologies but rather a lack of creativity in profitably deploying them, a belief that justifies the black box approach to modeling technologies adopted in this paper.

The stag hunt has been used as a model of various social coordination problems (for example, see \cite{PlayingFair, SocialContract, Skyrms, Medina}). Coordination games have also been used as a model of telecommunication networks and the Internet \cite{Economides1}. This paper's thesis should, therefore, not come as a surprise to economists and social scientists (and it should not come as surprising even to some computer scientists; for example, see \cite{Feamster}). But what is particularly appealing about this thesis is that it suggests a course of action that radically departs from the existing innovation efforts to change the Internet architecture. 

The course of action I propose is that innovation should be assisted by means of {\em social mechanisms} or {\em institutions} that should work in tandem with the emerging technologies toward their adoption. The term {\em institution} is to be understood precisely as a social mechanism whose goal is to coordinate activity among the members of a population \cite{Schotter}. Examples of institutions are the traffic code, markets, or regulations. That is not to say that I propose to effect innovation through markets or regulations: Both mechanisms are important institutions deeply rooted in the social structures our societies are based upon, however, casting the problem of innovating at the core of the Internet as a market design or as a regulatory intervention problem does in my view take full effect of the idiosyncratic characteristics of this adoption environment. 


For example, attempting to intervene in the adoption environment through direct regulation is problematic for a variety of reasons: First, it requires from engineers to convince the lawmaking authorities of nations with significant Internet presence of the importance of emerging technologies. Second, it requires solving a much harder international coordination problem among the lawmakers. Third, it is well-known that regulations are not equally effective in the nation-states with significant Internet presence. Finally, this process has to be repeated every time innovation is desirable.

\subsection{Our contributions}

Our first contribution in this paper is to propose two {\em coordination mechanisms,} that is, social institutions, that facilitate selection of the superior equilibrium in a stag hunt. Our use of the term `mechanism' suggests keen relevance of our approach to {\em mechanism design,} and although there are many similarities, at least at the conceptual level, there are also important differences the most important of which is that we are not attempting to entice a truthful revelation of preferences among outcomes (as is the typical situation in the design of mechanisms and their theoretical analysis), but rather to incur an {\em equilibrium transition} from an inferior to a superior outcome, assuming that the incentive structure (and, therefore, the respective preferences) is well-known. The first mechanism we propose relies on the presence of an external to the players entity that can assume the role of an {\em insurance carrier} having the capability to insure costly efforts while the second mechanism, one based on an {\em election}, assumes the ability to enforce voluntary player commitments. We note that these mechanisms are not specific to the architectural evolution problem that motivated our investigation, but are rather generic and apply to a broader set of coordination problems.

Our second contribution is a game-theoretic analysis of the aforementioned mechanisms. To that end, we leverage a variety of techniques, in particular, Nash equilibrium theory, the game-theoretic solution concept based of weak and strict iterated dominance, as well as a solution concept that we devise based on the notion of {\em incremental deployability} that originated in the networking literature to (colloquially) mean that a technology provides benefits to its early adopters vis-\'a-vis the incumbent environment in which it competes. In this paper, we rigorously formalize incremental deployability in a manner that renders this concept closely related to the notion of the {\em sink equilibrium,} a game-theoretic solution concept proposed by Goemans {\em et al.} \cite{Goemans}. 

The sink equilibrium is typically, if not exclusively, used in the literature based on a notion of dynamics characterized by better or profitable best responses. Such dynamics give rise to the pure Nash equilibrium as their equilibrium solution concept. In the interest of analyzing coordination mechanisms (in particular, the election mechanism) we develop a solution concept characterized by a notion of dynamics wherein evolution is allowed to ``drift'' according to unilateral player deviations that are not harmful giving rise to a new equilibrium solution concept in pure strategies. We call this solution concept a {\em strongly maximal equilibrium}. Such equilibria are necessarily pure Nash equilibria, but the converse is generally false. A distinguishing feature of this equilibrium concept is that strongly maximal equilibria can manifest in equilibrium classes or {\em clusters}. We argue that such an equilibrium concept, is as, if not more, plausible, than the pure Nash equilibrium. 

Economists as well as, more recently, algorithmic game theorists have been concerned with a class of games wherein better and best response dynamics necessarily give rise to pure Nash equilibria as the outcome of long-term asymptotic behavior. Such games are known as {\em weakly acyclic games} in the sense that, although cyclic behavior can manifest in their dynamics, such behavior is necessarily transient. In the interest of analyzing the election mechanism, we introduce a class of games that we call {\em weakly ordinally acyclic} capturing the aforementioned strongly maximal equilibria as the outcome of long-term asymptotic dynamics; we show that the election mechanism applied to the stag hunt induces a game that falls into this class. We finally note that, in the course of our analysis, we prove that the stag hunt is a {\em weakly (ordinally) acyclic game,} a result that, to the extent of our knowledge, has not appeared before in the literature.

\subsection{Overview of the rest of the paper}

The purpose of the next section (Section \ref{motivation}) is to provide further motivation for the problem we are tackling in this paper. Section \ref{coordination_mechanisms} introduces the concept of a coordination mechanism, an institution aiming to facilitate selection of the superior equilibrium in a stag hunt, first in abstract terms and then by instantiating it with particular examples that seem to fulfill the definition of such a mechanism. Naturally, whether an institution qualifies to be a coordination mechanism depends on the analytical method employed to predict the effect of that institution on the incentive structure of the adoption environment. Various methods can be employed to that effect, which are presented in Section \ref{analytical_methods}. The goal of Sections \ref{on_maximality} and \ref{stag_hunt_is_weakly_acyclic} is to add further support to the analytical approach being pursued in this paper, while the (aforementioned) examples of plausible coordination mechanisms, namely, a mechanism based on insuring the investments associated with cooperative strategies and a mechanism based on an election, are analyzed in Section \ref{Analysis}. Section \ref{discussion} places game-theoretic mechanisms in a broader context, relating our approach to tolls in road traffic transportation networks. Related work is further discussed in Section \ref{other_rel_work} and Section \ref{conclusion} concludes this paper.

\section{Motivation}
\label{motivation}

Technology (in contrast to natural or biological phenomena) is something humans design and, therefore, we would reasonably expect to know precisely how it works. That may indeed be true for some technologies such as processors but it is certainly not true for all technologies. I would say that the Internet belongs to the category of those technologies of which we are trying to understand the principles of their operation in the same manner that we are trying to understand the principles of operation of natural or biological phenomena. This paper is motivated, in part, by an effort to understand in a rigorous manner the nature of the Internet, how it evolves, whether we should intervene in its evolution, and how we might attempt such an intervention.

\subsection{The ossification of the Internet}

The motivation for this paper is a phenomenon that manifests at the core architecture of the Internet (that is, the TCP/IP architecture), namely, that despite significant attempts from the networking community to effect change, this core architecture is defiant of these efforts. Quoting Ratnasamy {\em et al.} \cite{RSM}: ``In the early days of the commercial Internet (mid 1990's) there was great faith in Internet evolution. \ldots The remarkable success of the Internet surpassed our wildest imagination, but our optimism about Internet evolution proved to be unfounded. \ldots Thus the ISPs which where once thought the agents of architectural change are now seen as the cause of the Internet impasse \ldots'' The core architecture of the Internet has facilitated significant innovation at the link layer (e.g., the rapid transition to 3G and 4G systems in wireless and the rapid transition to optical technologies for fixed access) as well as the application layer (e.g., Google, Facebook, BitTorrent), however, it is increasingly being held responsible for stifling the emergence of radical new applications.\footnote{\url{http://www.guardian.co.uk/technology/2012/apr/29/internet-innovation-failure-patent-control}}\\

The Internet started out as a research experiment whose stunning success led to an expansion of unprecedented scale. The Internet grew precipitously without giving room to the community for introspection on the design choices that had been made; brilliant though as these design choices were in many respects but confined to research experimentation in others, the designers could not do more than observe a research experiment escape the lab to weave into our societal fabric.

Users and operators alike of the infrastructure that eventuated are mounting calls for significant change toward improving performance, quality of service, availability, and security of the services offered by the infrastructure. However, the Internet has been {\em ossified} and despite significant efforts the architecture hasn't changed in decades. This, perhaps surprising and, to a large extent, inexplicable, ossification is the motivation for this work. We cannot change the basic architecture of the Internet like we can change any other technology; in fact, by and large, the Internet does not qualify to be a technology in the commonly understood sense of this term; atop the technical design has evolved an institutional structure that complements and is complemented by this design, a structure that we cannot turn a blind eye on whether to effect architectural change or to foster the growth of new technologies. 

As another case-in-point (on par with secure routing) that effecting changes to the core architecture of the Internet is a matter that transcends technical challenges I use the adoption of IPv6. This is an emerging technology operating at the {\em network layer}. The network layer's basic functions are addressing (that is, identifying hosts with an address so that they are reachable from other hosts) and routing (that is, establishing communication paths between hosts so that data can be forwarded from source to destination). The predominant version of IP operating in the Internet is IPv4 whose address space is scarce, however, and has already been depleted.\footnote{Various methods have been developed and deployed to address this limitation such as {\em network address translation,} however, these methods are heavily criticized as violating basic architectural principles.} IPv6 provides a significantly larger address space, but its adoption is progressing at a disappointingly slow pace despite concerted efforts by various players to expedite its diffusion.


\subsection{Efforts to innovate}

The difficulty, or perhaps impossibility, of effecting change at the core of the Internet architecture had generated much frustration at the networking research community. Quoting an unpublished position statement authored by Jennifer Rexford in 2004\footnote{http://www.cs.princeton.edu/~jrex/position/disrupt.txt} that has since been influential in my thinking: ``In the past several years, the networking research community has grown increasingly frustrated with the difficulty of effecting substantive change in the Internet architecture.  This frustration is responsible, in part, for a shift in focus to research in ``green-field" environments  such as overlays, peer-to-peer, sensor networks, and wireless networks.  Although these topics are compelling and important in their own right, I would argue that the underlying Internet infrastructure, IP protocols, and operational practices are plagued with serious problems that warrant  significant research attention from our community.''

Jennifer Rexford then goes on to argue in the same position statement that: ``However, having revolutionary impact on the Internet architecture is
extremely challenging because of the commercialization of the network
and the large installed base of legacy equipment running the existing
protocols.  As a result, many studies of the Internet's ``underlay"
focus on collecting and analyzing measurement data \ldots  These
kinds of studies are certainly interesting and important \ldots Yet, we
need to find ways to effect more fundamental change in the years ahead.''

The networking community's efforts toward effecting fundamental change in the Internet have been channeled toward two distinct directions. Computer scientists have become accustomed to attributing the innovation slump, on one hand, to the lack of appropriate technical characteristics in the existing technologies and attempt to counterbalance such limitations through devising technologies that are easier to implement, have better performance, and have better reliability and security, and that are incrementally deployable. However, these efforts have been to no avail. On the other hand, the community is creating large scale experimentation environments capable of even attracting a real user base such as the Global Environment for Network Innovations (GENI), an effort driven by the belief that adoption of proposals to effect change is hindered by the inevitable lack of deployment experience for untested in the field ideas.\footnote{An influential argument for building such large scale testbeds was articulated by \cite{Impasse}.} Although deployment experience is a necessary condition for the widespread adoption of a technology to take place, I argue that it is by no means sufficient; in this paper I proffer a radically different to the problem at hand perspective, and propose an (also radically different) means by which it can be addressed.

I should note, however, that the networking community is not oblivious to the fact that effecting architectural innovation is a matter that transcends technical specifications. In an influential paper, Anderson {\em et al.} \cite{Impasse} emphasize that ``In addition to requiring changes in routers and host software, the Internet's multiprovider nature also requires that ISPs jointly agree on any architectural change. [changing paragraph] The need for consensus is doubly damning: Not only is reaching agreement among the many providers difficult to achieve, attempting to do so also removes any competitive advantage from architectural innovation.''

Much like deployment experience is a necessary condition for adoption of an emerging technology to take place, {\em agreement} on the benefits of an emerging technology is also a necessary condition for adoption, and by no means sufficient. As Olson \cite{Olson} lucidly argues in a generally acceptable (and widely celebrated) viewpoint (that has its critics however), groups aiming to take collective action do not form through mere agreement of their members on the possible benefits of furthering their objectives or their members' likemindedness; incentives play a significant role in the formation of such groups. Olson's perspective forms the basis of this paper's thesis.

\section{Coordination mechanisms}
\label{coordination_mechanisms}

In this section, we present the main idea introduced in this paper, namely, that of a ``coordination mechanism,'' first in abstract terms, which we then instantiate using particular examples. In the setting of Internet architecture, these mechanisms could be employed to overcome the deployability barriers that innovative architectures whose success depends on positive externalities face. In this setting, coordination mechanisms correspond to a radical departure from ongoing innovation efforts. However, we should note that the application span of such mechanisms extends beyond the Internet to other endeavors facing a similar incentive structure. In fact, as argued in the sequel, the notion of a coordination mechanism is, in a sense, already being used in practical applications such as {\em crowdfunding} with great success (although viewing crowdfunding as a coordination mechanism in the sense the term is used in this paper is a novel, to the extent of our knowledge, perspective in the corresponding literature). Let us, however, start with preliminaries on strategic games.

\subsection{Preliminaries}

Let us begin with the definition of games in (finite) {\em strategic form} (also called a {\em normal form}). To define a game in this form, we need to specify the set of players, the set of strategies available to each player, and a utility function for each player defined over all possible combinations of strategies that determines a player's payoffs. Formally, a  strategic-form game $\Gamma$ is a triple $$\Gamma = (I, (S_i)_{i \in I}, (u_i)_{i \in I}),$$ where $I$ is the set of players, $S_i$ is the set of pure strategies available to player $i$, and $u_i: S \rightarrow \mathbb{R}$ is the utility function of player $i$ where $S = \times_i S_i$ is the set of all strategy profiles (combinations of strategies). Let $n$ be the number of players. We often wish to vary the strategy of a single player while holding other players' strategies fixed. To that end, we let $s_{-i} \in S_{-i}$ denote a strategy selection for all players but $i$, and write $(s'_i, s_{-i})$ for the profile
\begin{align*}
(s_1,\ldots,s_{i-1}, s'_i,s_{i+1},\ldots, s_n).
\end{align*}
A pure strategy profile $\sigma^*$ is a {\em Nash equilibrium} if, for all players $i$,
\begin{align*}
u_i(\sigma_i^*, \sigma_{-i}^*) \geq u_i(s_i, \sigma_{-i}^*) \text{ for all } s_i \in S_i.
\end{align*}
That is, a Nash equilibrium is a strategy profile such that no player can obtain a larger payoff using a different strategy while the other players' strategies remain the same. If unilateral deviations from a Nash equilibrium are necessarily harmful, the Nash equilibrium is called {\em strict,} and it is called {\em weak} otherwise. The previous definitions concern Nash equilibria in pure strategies, however, we note that Nash equilibria admit an analogous definition assuming players may use {\em mixed strategies,} that is, probability distributions over their respective pure strategy sets. Nash equilibria in mixed strategies are known to always exist in strategic form games, in contrast to pure Nash equilibria.

\subsection{Definition of the stag hunt}

Skyrms \cite{Skyrms} lucidly defines the stag hunt as follows: ``Let us suppose that the hunters [in a group] each have just the choice of hunting hare or hunting deer. The chances of getting a hare are independent of what others do. There is no chance of bagging a deer by oneself, but the chances of a successful deer hunt go up sharply with the number of hunters. A deer is much more valuable than a hare. Then we have the kind of interaction that is now generally known as the stag hunt.''

In the setting of architectural evolution, deer is the emerging technology, hare the incumbent, and the hunters are the potential adopters. That the chances of getting a hare are independent of what others do reflects the often realistic assumption that the incumbent neither benefits nor suffers from adoption of the emerging technology. That there is no chance of bagging a deer by oneself reflects that the emerging technology is not incrementally deployable, that the chances of a successful deer hunt go up sharply with the number of hunters reflects positive externalities, and that a deer is much more valuable than a hare reflects the superiority of the emerging technology. 

The stag hunt has two Nash equilibria in pure strategies, namely, universally adopting the superior cooperative strategy and universally defecting from using that strategy. The manifestation of the latter equilibrium is known as a {\em coordination failure}. The stag hung also has a mixed Nash equilibrium providing payoffs inferior to those of either of the pure Nash equilibria, deemed unstable, and generally ignored in the literature.

To illustrate the model, let us give an example on the problem of effecting architectural innovation in the Internet infrastructure that is organized in {\em autonomous systems,} which are administratively independent organizations offering Internet connectivity to endusers. To provide this function autonomous systems interconnect in a rather general topology (which satisfies certain statistical properties such as the power law distribution). We may represent this topology as a graph, which at the moment has approximately $50,000$ nodes.

In a typical adoption environment (in particular, for {\em network-layer} innovation), players are the administrative authorities of these autonomous systems---we use $i$ as the running index of a player. Each player has a choice of either adopting (strategy $A$) or defecting (strategy $D$). The payoff of an adopter depends on the connected component of adopters in the graph to which the adopter's autonomous system belongs, and this payoff goes up sharply with the size of this component. The payoff of a player who defects is zero. If the size of an adopter's component is small enough, the adopter's payoff may be negative (meaning that the deployment effort costs more than the benefit).

Denoting a combination of the players' choices (strategy profile) by $s$, writing this from the perspective of player $i$ as $(s_i, s_{-i})$ where $s_{-i}$ is a combination of strategies for all players except $i$, and denoting the payoff function of player $i$ as $u_i(\cdot)$, the previous discussion can be captured in the following simple formula:
\[
  u_i(s_i, s_{-i}) = \left\{ 
  \begin{array}{l l}
    \beta_i(s) - \gamma_i, & \text{$s_i=A$}\\
    0, & \text{$s_i=D$.}\\
  \end{array} \right.
\]
where $\beta_i(s) - \gamma_i$ is the net benefit of adoption (which may be negative depending on the outcome of the game) and $\gamma_i > 0$ is the deployment cost.

We provide below a formal definition of the stag hunt that we will use in the analysis in following sections; to the extent of our knowledge this is the first formal definition of this game (although it closely reflects the aforementioned description of \cite{Skyrms}).

\begin{definition}[Stag hunt]
We say that $\Gamma = (I = \{1, \ldots, n\}, (S_i)_{i \in I}, (u_i)_{i \in I})$ is a {\em stag hunt} if, for all $i \in I$, $S_i = \{A, D\}$ (that is, all players have a common pair of strategies) and the payoff of choosing $A$ (hunting deer) to each player that hunts deer is an increasing function of the players that hunt deer, whereas the payoff of choosing $D$ (defecting from hunting deer and hunting hare) is a constant, say, $c$. Furthermore, letting $A^n$ denote the profile in which every player hunts deer, we have that, for all $i \in I$, $u_i(A^n) > c$. Finally, all unilateral deviations from $D^n$ (that is, from the profile where everyone defects and hunts hare) are harmful.
\end{definition}

The last requirement that unilateral deviations from $D^n$ are harmful can be understood as follows: To hunt deer a player expends resources (such as his or her energy) and unless the hunt is successful these resources are wasted. In the setting of network evolution, the resources being wasted would typically correspond to an investment that does not pay off. An example of a two-player stag hunt is the following game:
\begin{align*}
\left(\begin{array}{c c}
(10, 10) & (-1, 0)\\
(0, -1) & (0, 0)\\
\end{array}\right).
\end{align*}
If both players cooperate in choosing the superior strategy each earns a payoff of $10$ whereas the uncooperative strategy yields a payoff of $0$ irrespective of what the other player does. If a player chooses to cooperate but the other player doesn't, the former earns a negative payoff.

\subsection{Definition of coordination mechanism}

Predicting the outcome of a stag hunt is a notoriously hard problem that can be formalized as an {\em equilibrium selection problem} provided we accept the Nash equilibrium as the starting point of the analysis, a rather common assumption in game theory with few exceptions. A fundamental assumption in most attempts to provide a rigorous foundation for game theory since its inception is that players attempt to {\em maximize utility}. In the stag hunt there is a unique outcome where utility is maximized uniformly across all players that corresponds to the selection of the superior equilibrium. However, experience suggests that the inferior equilibrium may also manifest in practice, an observation that is supported by empirical evidence in experimental settings. 

The manifestation of the inferior equilibrium agrees with Nash equilibrium theory, however, solving the equilibrium selection problem amounts to more. Interestingly, as discussed in a related Wikipedia entry,\footnote{\url{http://en.wikipedia.org/wiki/Risk_dominance}} in their seminal attempt to provide a solution to this latter problem, Harsanyi and Selten \cite{HS} identified selection of the superior equilibrium as corresponding to a ``rational decision'' but later Harsanyi retracted this position and identified selection of the inferior equilibrium as the rational choice. We note, however, that what constitutes rational choice in such environment is an unresolved problem that warrants further investigation. The mechanisms we present in this paper are designed to eschew ambiguities that arise owing to multiplicity of equilibria (or more generally, game-theoretic solutions). Let us, therefore, define these mechanisms in more precise terms.

We define a {\em coordination mechanism} broadly to be a {\em social institution} that induces the superior cooperative outcome in a stag hunt (called the {\em basis game}), and, therefore, averting a coordination failure, by means of merely adding strategies in the basis game, without affecting its payoff structure. That is, to qualify as a coordination mechanism, a social institution must not only induce cooperation but it must also be based on the voluntary participation of the players, who may choose to opt out if they so desire. Of course players opting in have to abide by the rules of the coordination mechanism---various situation-specific approaches can be taken to that effect. We note that whether an institution qualifies to be a coordination mechanism depends on the analytical method that is used to predict the outcome of the environment that is enhanced by the mechanism. In the sequel, we discuss a variety of methods that can be used for the analysis, drawing on them to analyze particular instances of coordination mechanisms, which we outline next.

\subsection{Examples of plausible coordination mechanisms}

In the rest of this section, we provide two examples of mechanisms that seem to be plausible candidate institutions meant to induce the cooperative outcome in a stag hunt. In subsequent sections, we prove that these mechanisms qualify as coordination mechanisms.\\

\noindent
The mechanisms' efficacy rests on the following, common to both mechanisms, assumptions: 
\begin{itemize}

\item[{\bf A1.}] It is possible to enforce {\em commitments}. 

\item[{\bf A2.}] Using the cooperative strategy requires an {\em investment,} and a third party can measure the cost of individual investment. 

\item[{\bf A3.}] A third party can verify ex post whether a particular player used the cooperative strategy.

\end{itemize}

In Internet architecture, an investment could correspond to the purchase of equipment and the implementation of technological standards, for example. Enforcing commitments is feasible across national borders through international commerce and contract law, while whether the second and third assumptions are satisfied depends on the nature of the emerging technology. Other mechanisms can also be employed toward the end of enforcing commitments however. For example, instead of relying on international contract law, a particularly compelling alternative to that end is to rely on nothing more than the institutional and organizational structure and the technological capabilities of the global Internet including the service providers, their business relationships, and the capabilities of their networks. For example, detecting infringement of commitments is possible through {\em peer review} and to punish infringement a coalition of service providers could {\em rate limit} those who infringe. The benefit of such design would be to decouple the Internet's evolution from international politics, and meeting the ensuing challenges is a fruitful direction for future work.

\subsubsection{An insurance mechanism}

The first mechanism is based on observing that adoption of the cooperative strategy in a stag hunt is a {\em risky undertaking} (as the investment involved in adopting such a strategy will not reap benefits unless others cooperate) and consists of instituting a carrier that {\em insures} these investments. There exist various possibilities on the carrier's governance structure as well as the specific insurance policies that can be offered to potential adopters. Perhaps the simplest of all insurance mechanisms is to offer policies that cover amounts slightly exceeding the corresponding potential adopters' deployment costs, and it is precisely this case that we analyze in the sequel. Observe that by offering such insurance policies, we do not modify the payoff structure of the basis game, and therefore the insurance mechanism seems quite a plausible candidate to qualify as a coordination mechanism. In our analysis in later sections, we show that this mechanism is indeed a coordination mechanism and that universal cooperation can, in principle, manifest without any of the potential adopters purchasing insurance, a rather surprising finding. We note that the idea of using insurance to induce the cooperative outcome in a stag hunt has been previously proposed by Autto \cite{Autto}, but our analysis is more thorough and, in part, relies on a different analytical framework.

\subsubsection{An election mechanism}

Our second mechanism is based on the intuitive observation that {\em communication helps coordination,} a fact supported by ample experimental evidence (e.g.,~\cite{coord-comm}). However, in a group as large as the group of autonomous systems in the Internet, pairwise communication is untenable. But this does not preclude the possibility that indirect (more economical) communication can be as effective. For example, an interesting idea is to give players the opportunity to vote for or against adoption of the emerging technology as a means of either stimulating confidence that adoption is viable or renouncing adoption ahead of the investments. Such an idea is amenable to empirical testing, but, lacking precise models of the influence of signals (such as the outcome of the vote) on human behavior, predicting this mechanism's effect on adoption is a hard if not intractable problem. Our solution draws, however, on this idea and on the aforementioned feasibility assumptions.

We only consider in this proposal a simple election mechanism, the simplest we can think of that can be analytically shown to avert a coordination failure. According to this mechanism players have the option to vote for adoption, and once voting is complete, the mechanism outputs `1' if all players voted and `0' otherwise; the mechanism does not provision for negative votes and it does not disclose the identity of the players who voted. Such a mechanism being in place, players can condition their adoption decision on their own vote and on the outcome of the election. However, to have a predictable effect on adoption, payoffs must depend on the election outcome, and, to that effect, we introduce the following rule: Voting becomes a {\em commitment} to adopt if all players vote in which case defectors must suffer a penalty not smaller than their investment cost. The aforementioned feasibility assumptions ensure that such a mechanism can be implemented.  

We should note that, from a practical perspective, in adoption environments consisting of thousands of players the assumption of requiring that all players vote in favor of adoption to enforce a commitment is perhaps not realistic; for example, even if we assume that all players are rational, human errors, software errors (for example, in the software system that implements the election mechanism), or combinations thereof are not unlikely to manifest. However, commitments can be reasonably enforced as long as a large enough subset of players vote favorably without affecting the mechanism's efficacy: If a large enough subset of players adopt, then evolutionary forces alone inherent in the adoption environment are likely to induce themselves manifestation of the superior equilibrium of universal adoption (throughout the entire player population).

Our election mechanism is not a mere theoretical curiosity; it is akin to a {\em funding mechanism} (called {\em crowdfunding}) various platforms such as Kickstarter use to fund creative projects with great success. In this mechanism, possible contributors interested in seeing the idea of a project proposal come to reality can {\em pledge} an amount to a particular project proposal and if the proposal meets a target amount, the amounts pledged are contributed to the project; such a mechanism is evidently akin to voting in favor of an emerging technology and committing to adopt if the election outcome is positive; as noted in Kickstarter's homepage: ``All-or-nothing funding might seem scary, but it's amazingly effective in creating momentum and rallying people around an idea.'' Although it is not immediately clear whether the incentive structure of the environment crowdfunding platforms face can be modeled after a stag hunt, it is certainly a problem of {\em collective action} (a class of problems that, in a certain broad definition, encompasses stag hunts). Coordination mechanisms in these related environments warrant further independent scrutiny and we believe the analytical ideas introduced in this paper can make an important contribution to that end.\\

Let us briefly discuss the approach we will take in analyzing the election mechanism. Observe to that end that, as mentioned earlier, without the election mechanism universal defection is a strict equilibrium. However, once this mechanism is instituted it is not harmful for any player to vote in favor of adoption, assuming, for example, players keep in mind the reasonable strategy of defecting if the outcome of the vote is negative. Observe further that if all players think in this way and decide to vote, then, by the design of the mechanism, all players commit to adopt the cooperative strategy (implying emergence of the cooperative outcome). This rather reasonable intuitive analysis does not admit a straightforward game-theoretic justification; in fact, even if the election mechanism is instituted, universal defection remains a pure Nash equilibrium, albeit weak (and, therefore, of less predictive value than its strict counterpart). In the sequel, one of our main goals is to place squarely the aforementioned informal intuitive observations as of the election mechanism's efficacy in the game-theoretic apparatus, and many of our results are obtained in this vein.

\section{Analyzing coordination mechanisms}
\label{analytical_methods}

To analyze coordination mechanisms, the approach we take is naturally {\em game-theoretic}. Our analysis is couched in terms of well-known solution concepts such as the Nash equilibrium and dominance solvability but we also depart from these established analytical approaches in the interest of developing a solution method that is motivated by the empirical intuition of researchers and practitioners in the networking community. We demonstrate that this latter analytical approach bears keen relevance to a game-theoretic solution concept formalized by theoretical computer scientists.

\subsection{Nash equilibrium analysis}

Central in game theory is the concept of {\em mixed strategies}. To the extent of my knowledge, the concept was introduced by von Neumann in proving the minimax theorem. Furthermore, Nash's celebrated result is that mixed strategy Nash equilibria always exist in strategic games. These results are certainly intuitively appealing, however, the empirical evidence that players facing strategic environments randomize as they select their corresponding strategies is sparse. 

There is a rich literature on mixed Nash equilibria; most game-theoretic solution concepts are couched in this framework. One of the fundamental problems Nash equilibrium theory faces is related is to the possible multiplicity of solutions; for example, coordination games are well known, to some extent by definition, to be amenable to this phenomenon. To the purpose of analyzing collective action problems (according to a certain broad definition that encompasses the stag hunt, the game of primary interest in this paper), Medina \cite{Medina} introduces the method of {\em stability set analysis} extending the equilibrium selection theory of Harsanyi and Selten \cite{HS}. Stability set analysis is based on the assumption that players introspect about the game they are facing assuming either they or their opponents are in the position to randomize their selection of strategy, what may be a less elementary to the architectural innovation problem motivating this paper assumption.

The literature on strategic games in pure strategies is sparse by comparison. But the engagement of computer scientists in game theory through their {\em algorithmic lens} is promising to put new life into the, what seems more natural, discrete nature of strategic games. Such a shift is, in fact, already noticeable in the literature on the intersection between economics and computation, the introduction of the concept of the {\em sink equilibrium} being a case-in-point \cite{Goemans}.

I will not take a general stance in favor of one or the other type of approach to studying strategic games in the general setting (both perspectives can yield useful insights in all likelihood), however, I believe that in the more narrow setting of interest in this paper, namely, the problem of evolution at the core of the Internet architecture, the concept of mixed strategies is less elementary; as a case-in-point mixed strategies do not play a significant role in the study of coordination games, the model upon which my study of the problem of architectural evolution is based.

Viewed as a solution concept, the pure Nash equilibrium is less attractive than its mixed counterpart as it does not satisfy the existence property; but, if they exist, pure Nash equilibria are generally deemed as plausible solutions. We should note right upfront, however, that Nash equilibrium theory alone may provide useful predictions depending on the coordination mechanism being employed; for example, if a coordination mechanism induces a game with a unique Nash equilibrium, Nash equilibrium theory would predict this outcome as the game's only possible solution.

\subsection{Dominance solvability analysis}

A basic rationality postulate in game theory is that a player will never play a (strictly) ``dominated'' strategy. There are two notions of dominance, namely, {\em strict} (also sometimes called {\em strong}) dominance and {\em weak} dominance. We introduce these definitions below.

\begin{definition}[Strict dominance]
Let $\Gamma = (I, (S_i)_{i \in I}, (u_i)_{i \in I})$ be a game in strategic form, and consider player $i \in I$. We say that pure strategy $s_i$ {\em strictly dominates} the pure strategy $s'_i$ if if for all $\sigma_{-i} \in S_{-i}$, we have that
\begin{align*}
u(s_i, \sigma_{-i}) > u(s'_i, \sigma_{-i}).
\end{align*}
\end{definition} 

\begin{definition}[Weak dominance]
Let $\Gamma = (I, (S_i)_{i \in I}, (u_i)_{i \in I})$ be a game in strategic form, and consider player $i \in I$. We say that pure strategy $s_i$ {\em weakly dominates} the pure strategy $s'_i$ if if for all $\sigma_{-i} \in S_{-i}$, we have that
\begin{align*}
u(s_i, \sigma_{-i}) \geq u(s'_i, \sigma_{-i})
\end{align*}
with the inequality being strict for at least one $\sigma_{-i}$.
\end{definition} 

Strict dominance is, in a sense, compatible with the Nash equilibrium as it is possible to show that a strictly dominated strategy will never appear in a Nash equilibrium. Therefore, eliminating strictly dominated strategies from a game is a means of simplifying a game without running the risk of eliminating Nash equilibria unlike eliminating weakly dominated strategies.

The idea of eliminating dominated strategies from a game gives rise to the solution concept of {\em iterated dominance,} which appears in the literature in two variations, one {\em strict} and one {\em weak} according to the elimination criterion being used. Iterated dominance proceeds in rounds where in each round dominated strategies are eliminated until no strategies that can be eliminated remain. If the outcome of this process is a unique strategy for each player, the respective game is called {\em dominance solvable}. If a game is dominance solvable by strict iterated dominance, the respective strategy profile can be shown to be the unique (pure) Nash equilibrium of the game. The outcome of iterated weak dominance depends, unlike its strict counterpart, on the order by which weakly dominated strategies are eliminated. Despite this limitation weak iterated dominance is generally acknowledged as a credible means of drawing conclusions on a game's possible outcome.

\subsection{Evolution by incremental deployability and maximal states}

The term ``incremental deployability'' originates in the networking literature. It does not admit a rigorous definition but is colloquially used in the following sense: An emerging technology is called incrementally deployable against an incumbent technology if early adopters of the emerging technology benefit \cite{Clean-Slate}. In the rest of this section, we devise a solution concept of strategic games using this notion. To that end, we first introduce some preliminary definitions.

Given a game $\Gamma$ in strategic form, we define its {\em deployment graph} as the directed graph $G$ whose vertices correspond to $\Gamma$'s strategy profiles and whose arcs correspond to a notion of ``feasible'' unilateral deviations. Therefore, the set of vertices of $G$ is the direct product $\times_i S_i$. The notion of feasibility may admit a variety of plausible definitions. In this paper, we explore two such definitions, namely, we either call a unilateral deviation feasible if it is (strictly) profitable for the corresponding player that deviates (switches strategies) or if it is not harmful. For example, the deployment graph of the aforementioned stag hunt is shown in Figure \ref{dgsh}.

\begin{figure}[tb]
\centering
\includegraphics[width=7cm]{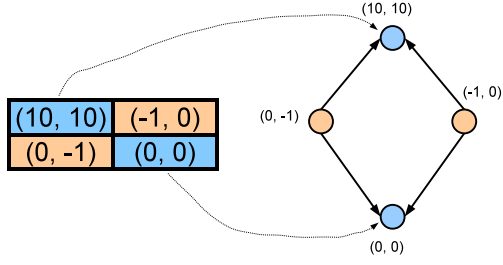}
\caption{\label{dgsh}
Example of deployment graph of a (two player) stag hunt.}
\end{figure}

We are now ready to provide a rigorous definition of incremental deployability: We say that a strategy profile (state) $\sigma$ of $\Gamma$ is {\em incrementally deployable} against strategy profile (state) $\sigma'$ if there exists a directed path from $\sigma'$ to $\sigma$ in $\Gamma$'s deployment graph $G$. To understand this definition in the particular setting of architectural innovation we consider, we may think of $\Gamma$ as capturing the incentive structure of the adoption environment where possibly multiple technologies (corresponding in this particular setting to strategies) compete for adoption. Then a state (a mix of technologies) is incrementally deployable against another state, if there exists a sequence of beneficial ``moves'' players of the adoption environment can take unilaterally, where a move corresponds to a player abandoning one technology in favor of another, to effect such change in the technology mix.

The following question seems to be in order given the aforementioned definitions: In an environment where evolution is driven by the process of incremental steps we have defined, which states, if any, would seem natural to emerge in the ``steady state'' as evolutionary forces have settled? Although we do not specify a particular mechanism by which evolution should take place, it is natural to assume that if a state $s$ is incrementally deployable against another state $s'$ but $s'$ is not incrementally deployable against $s$, $s'$ cannot manifest other than as a transient state, an assumption in agreement with the colloquial notion of incremental deployability as that is used in the networking community. The following definitions rigorously capture this intuition.

To formalize our solution concept in an environment where evolution is driven by incremental deployability we rely on {\em order-theoretic} parlance. Order theory has its origins in {\em decision theory} (initially formalized by von Neumann and Morgernstern \cite{GameTheory}) and it is typically used as an analytical tool in the interest of studying the decision making behavior of individual players in strategic environments (see also \cite{MWG}). In this section, we use order theory to formalize collective behavior.

\subsubsection{Preliminaries on binary relations}

Let $S$ be a nonempty set. A subset $R$ of $S \times S$ is called a {\em binary relation} on S. If $(s,s') \in R$, we write $sRs'$, and if $(s,s') \not\in R$, we write $\neg sRs'$. $R$ is called {\em reflexive} if $s R s$ for every $s \in S$. $R$ is called {\em complete} if for each $s,s' \in S$ either $sRs'$ or $s'Rs$. $R$ is {\em transitive} if for any $s, s', s'' \in S$ we have that $sRs'$ and $s'Rs''$ imply $sRs''$. $R$ is {\em symmetric} if, for any $s, s' \in S$, we have that $sRs' \Rightarrow s'Rs$ and {\em asymmetric} if, for any $s, s' \in S$, we have that $sRs' \Rightarrow \neg s'Rs$. Let $sPs' \Leftrightarrow sRs' \wedge \neg s'Rs$ and $sIs' \Leftrightarrow sRs' \wedge s'Rs$. Then $P$ and $I$ are also binary relations on S where $P \subseteq R$ and $I \subseteq R$. $P$ is called the {\em asymmetric} (or {\em strict}) part of $R$ and $I$ is called its {\em symmetric} part. 

\begin{definition}
A binary relation $\sim$ on a nonempty set $S$ is called an {\em equivalence relation} if it is reflexive, symmetric, and transitive. For any $s \in S$, the {\em equivalence class} of $s$ relative to $\sim$ is defined as the set
\begin{align*}
[s]_{\sim} = \{ \sigma \in S | s \sim \sigma \}.
\end{align*}
The collection of all equivalence classes relative to $\sim$, denoted as $S/_\sim$, is called the {\em quotient set} of $S$ relative to $\sim$, that is, $S/_{\sim} = \{ [s]_{\sim} | s \in S \}$.
\end{definition}

A well known result on equivalence relations is that for any equivalence relation $\sim$ on a nonempty set $S$, the quotient set $S/_{\sim}$ is a partition of $S$.

\begin{definition}
A binary relation $\succeq$ on a nonempty set $S$ is called a {\em preorder} on $S$ if it is transitive and reflexive.
\end{definition}

\begin{definition}
Let $S$ be an nonempty set, $\succeq$ a binary relation on $S$, and let $\succ$ be its asymmetric part. If $s \succ s'$ we say that $s$ {\em dominates} $s'$. An element $s^*$ of $S$ is called $\succeq$-maximal if it is {\em undominated}, that is, if for all $s \in S$, $\neg (s \succ s^*)$.
\end{definition}

An elegant proof of the following result can be found in the online manuscript Elements of Order Theory by Ok.\footnote{\url{https://files.nyu.edu/eo1/public/books.html}}

\begin{proposition}
\label{weoprtiueorituoirtu}
Let $S$ be a nonempty finite set, and let $\succeq$ be a preorder on $S$. Then there exists an element $s^*$ of $S$ such that $s^*$ is $\succeq$-maximal.
\end{proposition}

\subsubsection{Maximal states as a game-theoretic solution concept}

Let us first define deployment graphs in more rigorous terms. Recall that there are two distinct notions of such graphs, one where arcs correspond to (strictly) profitable unilateral deviations and one where are arcs correspond to unilateral deviations that are not harmful, which we will refer to as {\em strict deployment graphs} and {\em ordinal deployment graphs} respectively.

\begin{definition}[Strict deployment graph]
Let $S$ be the profile space of a normal form game $\Gamma$. We define the {\em improvement graph} $G(S, A)$ of $\Gamma$ as the directed graph that represents the players' (strictly) profitable unilateral deviations, that is, 
\begin{align*}
(s, s') \in A \Leftrightarrow \exists i : (s'_i, s_{-i}) = s' \mbox{ and } u_i(s') - u_i(s) > 0.
\end{align*}
\end{definition}

\begin{definition}[Ordinal deployment graph]
Let $\Gamma$ be a normal form game, and let $S$ be its profile space. We define the {\em advancement graph} $G(S, A)$ of $\Gamma$ as the directed graph that represents the players' weakly profitable unilateral deviations, that is, 
\begin{align*}
(s, s') \in A \Leftrightarrow \exists i : (s'_i, s_{-i}) = s' \mbox{ and } u_i(s') - u_i(s) \geq 0.
\end{align*}
If the above inequality is strict, we call the arc {\em positive}, and otherwise we call the arc {\em neutral}. Note that neutral arcs always come in {\em pairs} (that is, if $(s, s')$ is neutral, then $(s',s)$ is a neutral arc). 
\end{definition}

We are now ready to define our solution concept in two distinct variants according to the notion of a deployment graph that is appropriate in the modeling environment that is of interest.

\begin{definition}[Weakly maximal states]
Let $\Gamma$ be a normal form game, let $S$ be its profile space, let $G$ be its strict deployment graph, and order the vertices of this graph as follows:
\begin{align*}
s \succeq s' \Leftrightarrow \mbox{There exists a path in $G$ from $s'$ to $s$}.
\end{align*}
We call $\succeq$, the strict preference relation on $S$. We call the maximal elements of $(S, \succeq)$ the {\em weakly maximal states} of $\Gamma$.
\end{definition}

\begin{definition}[Strongly maximal states]
Let $\Gamma$ be a normal form game, let $S$ be its profile space, let $G$ be its ordinal deployment graph, and order the vertices of this graph as follows:
\begin{align*}
s \succeq s' \Leftrightarrow \mbox{There exists a path in $G$ from $s'$ to $s$}.
\end{align*}
We call $\succeq$, the drifting preference relation on $S$. We call the maximal elements of $(S, \succeq)$ the {\em strongly maximal states} of $\Gamma$.
\end{definition}

Both the strict and drifting preference relations are transitive, and since reflexivity is trivial to satisfy by a matter of definition, both types of preference relations are preorders. Proposition \ref{weoprtiueorituoirtu} then implies the following theorem.

\begin{theorem}
\label{existence}
Every strategic-form game is equipped with both weakly and strongly maximal states.
\end{theorem}

Therefore, as game-theoretic solution concepts in pure strategies, weakly and strongly maximal states satisfy existence, a property on par with the Nash equilibrium in mixed strategies.

We note that an alternative proof of existence of (weak and strong) maximal states can be couched in graph-theoretic terms using the well-known property that a ``directed acyclic graph'' (using the standard notion of acyclicity) is guaranteed to have a {\em sink} (i.e., a vertex without outbound arcs). To that end, we consider the quotient sets of the strict and drifting preference relations of a normal form game. Then, if we contract vertices belonging to the same member of the quotient set, we obtain a directed acyclic graph that, as mentioned above, is always equipped with a sink, which corresponds to a class of maximal elements (either in the strict or drifting preference relation).

\subsubsection{Maximal states and sink equilibria}

Maximal states are a game-theoretic solution concept defined in a ``static'' manner. In the sequel, we provide a dynamic interpretation of maximal states in a dynamic model of evolution using the well-studied concept of {\em Markov chains}. Let us first introduce some relevant background definitions.

A Markov chain is a sequence of discrete random variables taking values in some countable set, say $S$, called the {\em state space}. If the cardinality of $S$ is finite, we are referring to a {\em finite state} Markov chain. A Markov chain is called {\em homogeneous} is its behavior can be described in terms of {\em transition probabilities} between states: The probability of transitioning from one state, say $s \in S$, to another, say $s' \in S$, depends neither on history nor the time step but only on $s$ and $s'$. 

The ``solution concept'' in homogeneous finite state Markov chains is defined based on a notion of ``recurrency.'' In such a Markov chain, states are classified as being either {\em transient} or {\em recurrent}. A state is  recurrent if, starting from that state, the probability of returning to it is one, and it is transient otherwise. Finite state homogeneous Markov chains admit a graph-theoretic representation (cf. \cite{Gallager2}) in which states correspond to vertices and feasible transitions, i.e., transitions between states that may occur with positive probability, correspond to arcs. (Such a graph is typically referred to in the literature as the {\em transition graph}.) It is possible to show that recurrent states are independent on the precise values of transition probabilities (for as long as arcs correspond to transitions that may occur with strictly positive probabilities). That recurrent states are ensured to exist and that, starting from any state, a chain will eventually transition to a recurrent class of states (and remain there, thereafter) are well-known results in this theory.

With this background in mind, it seems natural to define a game-theoretic solution concept based on the notion of homogeneous finite state Markov chains, which is known in the literature on algorithmic game theory as the (aforementioned) {\em sink equilibrium} \cite{Goemans}. Sink equilibria are classes of recurrent states in a Markov chain where states correspond to the strategy profiles of a game in strategic form and where transition probabilities are defined based on better and best responses.  

It can be shown using standard characterizations of recurrent states (as can be found, for example, in \cite{Gallager2}) that the notions of maximal states and sink equilibria coincide (as long as the deployment and transition graphs coincide). This implies that maximality, as we earlier defined it using the static concept of incremental deployability, admits a rather appealing evolutionary interpretation in which players take turns and switch strategies according to the transition (deployment) graph. The outcome of this process can be either a pure Nash equilibrium, a collection (class) of equilibrium states, or more complicated cyclical behavioral patterns, as discussed in the sequel.

In closing, we note that the solution concept of the sink equilibrium is used exclusively, to the extent of our knowledge, in the literature as being synonymous with the notion of weak maximality. In the next section, we attempt a more thorough comparison between the notions of weak and strong maximality that we draw on later in the analysis of coordination mechanisms.

\section{Weak vs. strong maximality}
\label{on_maximality}

Weakly and strongly maximal states are rather different solution concepts. It is easy to see that weakly maximal equilibria coincide with pure Nash equilibria, however, although strongly maximal equilibria are necessarily pure Nash equilibria, there exist pure Nash equilibria that are not strongly maximal states. (We will come across this rather fine point again in the analysis of the election mechanism.) In this section, we argue that strongly maximal equilibria is an as plausible equilibrium solution concept in strategic games as their well-established weakly maximal counterparts (pure Nash equilibria). Our argument is couched in the setting of {\em ordinal potential games}.

\subsection{Potential games}

Potential games were introduced by Monderer and Shapley \cite{Potential} as strategic form games whose incentive structure is captured by a scalar potential function $P: S \rightarrow \mathbb{R}$ where $S$ is the space of strategy profiles of a game in strategic form. Although there are many variants, two important classes of potential games are the {\em ordinal potential games} and the {\em generalized ordinal potential games}. Let us, therefore, introduce these notions more formally: A function $P: S \rightarrow \mathbb{R}$ is an {\em ordinal potential} for a (finite) strategic form game $\Gamma$, if
\begin{align*}
u_i(s, \sigma_{-i}) - u_i(s', \sigma_{-i}) > 0 \Leftrightarrow P(s, \sigma_{-i}) - P(s', \sigma_{-i}) > 0, \forall i \in I, \forall s, s' \in S_i, \forall \sigma_{-i} \in S_{-i}.
\end{align*}
$\Gamma$ is called an {\em ordinal potential game} if it admits an ordinal potential. A function $P: S \rightarrow \mathbb{R}$ is a {\em generalized ordinal potential} for a strategic form game $\Gamma$, if
\begin{align*}
u_i(s, \sigma_{-i}) - u_i(s', \sigma_{-i}) > 0 \Rightarrow P(s, \sigma_{-i}) - P(s', \sigma_{-i}) > 0, \forall i \in I, \forall s, s' \in S_i, \forall \sigma_{-i} \in S_{-i}.
\end{align*}
$\Gamma$ is called an {\em generalized ordinal potential game} if it admits a generalized ordinal potential. These classes of potential games are distinct from each other: Ordinal potential games are necessarily generalized ordinal potential games, as is apparent from the definition, but the converse is not necessarily true. Monderer and Shapley show that the class of generalized ordinal potential games coincides with the class of (finite) games in strategic form having the {\em finite improvement property,} which is typically abbreviated as the FIP. A game in strategic form has the FIP if ``improvement paths,'' that is, sequences of unilateral better responses, are necessarily finite. They further show that games having the FIP are necessarily equipped with a pure Nash equilibrium.

\subsection{Strongly maximal equilibria}

In the sequel, we introduce an equilibrium solution concept in strategic form games in pure strategies, which we refer to as a {\em strongly maximal equilibrium,} that, in general, differs from the pure Nash equilibrium: Strongly maximal equilibria are necessarily pure Nash equilibria, but there exist pure Nash equilibria that are not strongly maximal. Noting that an example illustrating this point will follow, let us introduce the notion of a strongly maximal equilibrium more formally. 

To that end, let us consider the ordinal deployment graph (as defined earlier). In such a graph, we say that two states (strategy profiles), say $s$ and $s'$, {\em communicate} if deployment paths exist in both directions, that is, from $s$ to $s'$ as well as from $s'$ to $s$. We say that a set of states is a {\em strongly maximal equilibrium class} if all states in the set are pure Nash equilibria that pairwise communicate. We call the states in a strongly maximal equilibrium class {\em strongly maximal equilibria}. In the sequel, we prove that maximal states in ordinal potential games are necessarily strongly maximal equilibria. Since maximal states are guaranteed to exist in any game in strategic form, strongly maximal equilibria are guaranteed to exist in ordinal potential games. We also give an example of a generalized ordinal potential game that is not equipped with a strongly maximal equilibrium.

\begin{definition}
Let $\Gamma$ be a game in strategic form, let $S$ be its space of strategy profiles (states), and let $G$ be its ordinal deployment graph. Let us further consider the quotient set $S/_\sim$. We say that $G$ is {\em ordinally acyclic} if all arcs within each member of $S/_\sim$ are neutral.
\end{definition}

In the following theorem, we call a cycle {\em positive} if it contains at least one positive arc.

\begin{theorem}
\label{opg}
A game $\Gamma$ in strategic form is an ordinal potential game if and only if $\Gamma$'s ordinal deployment graph $G$ is ordinally acyclic.
\end{theorem}

\begin{proof}
Let us assume first that $\Gamma$ is an ordinal potential game, let $S$ be $\Gamma$'s space of strategy profiles (states), and let us further assume, for the sake of contradiction, that $\Gamma$'s ordinal deployment graph $G$ has a positive cycle, say consisting of vertices $s_1, \ldots, s_k$, which by the definition of the quotient set $S/_\sim$ these vertices belong to the same member of $S/_\sim$. Then, by the definition of of an ordinal potential games, it must hold that $P(s_k) \geq \cdots \geq P(s_1) \geq P(s_k)$, where $P(\cdot)$ is $\Gamma$'s ordinal potential function, and at least one of the previous inequalities is strict. But this is an impossibility since potential functions are single-valued. Therefore, all arcs within a member of the quotient set $S/_\sim$ are neutral, and, thus, ordinal potential games are necessarily ordinally acyclic.

Conversely, suppose $\Gamma$ is ordinally acyclic (that is, all arcs in each member of the quotient set $S/_\sim$ are neutral). We may then construct an ordinal potential function $P : S \rightarrow \mathbb{R}$ as follows: Consider the graph that results once we contract (using the standard graph-theoretic notion of a contraction in directed graphs) all neutral arcs in each member of $S/_\sim$. Call this the strict ordinal deployment graph, which is easy to prove that it is acyclic. Consider further a topological ordering of the strict ordinal deployment graph, say $(v_1, \ldots, v_K)$, and assign potential values to these vertices such that $P(v_K) > \cdots > P(v_1)$. Noting that each vertex $v_i, i =1, \ldots, K$ corresponds to a set of vertices of the ordinal deployment graph $G$, and assigning to each such vertex the potential value that has been assigned to the corresponding vertex of the strict ordinal deployment graph, it can be easily verified the potential function we have defined is an ordinal potential function.
\end{proof}

The previous theorem implies that all maximal states of an ordinal potential game are strongly maximal equilibria. Theorems \ref{existence} and \ref{opg} imply the following corollary.

\begin{corollary}
\label{existence_strong}
Every ordinal potential game is equipped with a strongly maximal equilibrium.
\end{corollary}

Let us finally give an example of a generalized ordinal potential game that is not equipped with a strongly maximal equilibrium. Consider, to that end, the following game:
\begin{align*}
\left(\begin{array}{c c}
(1, 0) & (2, 0)\\
(2, 0) & (0, 1)\\
\end{array}\right).
\end{align*}
Monderer and Shapley \cite{Potential} show that the previous game is not an ordinal potential game but it is a generalized ordinal potential game. Observe that the strategy profile $(2, 0)$ is the game's unique pure Nash equilibrium, however, it is easy to see that it is not strongly maximal.

\subsection{Plausibility of the strongly maximal equilibrium}

The Nash equilibrium in mixed strategies admits a characterization rendering it quite appealing in strategic environments where mixed strategies are a plausible modeling artifact. Norde et al. \cite{Norde} show that the mixed Nash equilibrium is characterized by three properties referred to as {\em existence}, {\em one-person rationality,} and {\em consistency}. It will become evident from the definition of one-person rationality and consistency that these are intuitive properties that {\em any} equilibrium concept should satisfy. (To some extent, the property of existence of an equilibrium, although desirable and pleasing, does not, in my opinion, affect the predictive value of a solution concept as argued by the abundance of non-equilibrium behavior in nature, biology, and society.) However, although the characterization of Norde et al. implies that any Nash equilibrium refinement must violate one of these properties within the entire class of normal form games, within narrower classes of games the set of mixed Nash equilibria is not necessarily {\em minimal} with respect to these properties \cite{Peleg}. In particular, within the class of ordinal potential games, the mixed Nash equilibrium is known not to be minimal \cite{Peleg}. In the sequel, we show that in this class of games strongly maximal equilibria satisfy one-person rationality and consistency (as well as existence as implied by Corollary \ref{existence_strong}).\\

Let $\mathcal{G}$ be a class of normal form games. A {\em solution concept} (or {\em solution rule}) in $\mathcal{G}$, $\phi : \mathcal{G} \rightarrow 2^{S}$, assigns to each normal form game $\Gamma \in \mathcal{G}$ a possibly empty subset of the profiles $S$ of $\Gamma$. A solution rule $\phi$ is said to have the {\em existence} property within the class $\mathcal{G}$ if, for all games $\Gamma \in \mathcal{G}$, there exists at least one profile selected by the rule, that is, $\phi(\Gamma) \neq \emptyset$. Let $\mathcal{G}_1$ be the class of one-person games within $\mathcal{G}$. A solution rule $\phi$ is said to have the {\em one-person rationality} property within the (broader) class $\mathcal{G}$ if, for all one-person games in $\mathcal{G}_1$, $\phi$ selects profiles whose payoff is maximum.

Let $\Gamma = (I, (S_i)_{i \in I}, (u_i)_{i \in I})$ be a normal form game within a class $\mathcal{G}$ of normal form games, let $s \in S$ where $S$ is the profile space of $\Gamma$, and let $J$ be a {\em proper subcoalition} of the player set $I$ of $\Gamma$, in that $J$ is a nonempty proper subset of the player set. The {\em reduced game} $\Gamma^{J, s} = (J, (S_j)_{j \in J}, (w_j)_{j \in J})$ of $\Gamma$ with respect to $J$ and $s$ is a normal form game whose player set is $J$. Each player $j \in J$ has the same strategy space $S_j$ as in $\Gamma$. Let $\mathcal{S} = \times_{j \in J} S_j$ be the profile space of the reduced game. The reduced game's payoffs are defined such that, for each profile $\sigma \in \mathcal{S}$, $w_j(\sigma) = u_j(\sigma, s_{-J})$. A solution rule is said to be {\em consistent} within a class $\mathcal{G}$ of normal form games, if, for all games $\Gamma \in \mathcal{G}$, all proper subcoalitions $J \subset I$ and all solutions $s^* \in \phi(\Gamma)$, we have that 
\begin{description}

\item[(i)] If players outside the subcoalition $J$ select strategies corresponding to $s^*$, then the resulting reduced game $\Gamma^{J, s^*}$ falls within the class $\mathcal{G}$.

\item[(ii)] The profile $s^*_J$ of the reduced game where players of the subcoalition $J$ also select strategies corresponding to $s^*$ is a solution of the reduced game  $\Gamma^{J, s^*}$.

\end{description}

In the proof of the following theorem, we need a lemma in whose statement and proof some definitions are in order. Given a (simple) path in the ordinal deployment graph, let us call the first vertex in the path that path's {\em tail} and the last vertex its {\em head}. Let us call a path in the ordinal deployment graph an {\em advancement path} if at least one of its respective arcs is a positive arc.

\begin{lemma}
\label{help_lemma}
Let $\Gamma$ be an ordinal potential game, let $S$ be its profile space, and let $\mathcal{W}$ be the set of all paths in $\Gamma$'s ordinal deployment graph. Then $s^* \in S$ is maximal if and only if there is no advancement path in $\mathcal{W}$ whose tail is $s^*$.
\end{lemma}

\begin{proof}
Suppose that there exists such a path whose tail is $s^*$ and let $s$ be its head. Then clearly $s \succeq s^*$. We claim then that, in fact, $s \succ s^*$, for if any path exists whose head is $s^*$ and whose tail is $s$, then concatenating the advancement path from $s^*$ to $s$ and the path from $s$ to $s^*$ gives a positive cycle, contradicting the assumption that $\Gamma$ is an ordinal potential game.

Suppose now that $s^*$ is not maximal. Then there exists $s$ such that $s \succ s^*$. We claim then that, in fact, the path from $s^*$ to $s$ is an advancement path. For if it is not, then it must be a neutral path (that is, a path all of whose arcs are neutral), which would contradict that $s \succ s^*$.
\end{proof}

\begin{theorem}
\label{qwerporeiroierurux}
Let $O(\mathcal{G})$ be the class of ordinal potential games. Then the solution rule $\phi : O(\mathcal{G}) \rightarrow 2^{S}$ that selects for each $\Gamma \in O(\mathcal{G})$, the set of strongly maximal equilibria of $\Gamma$ satisfies existence, one person rationality, and consistency.
\end{theorem}

\begin{proof}
Corollary \ref{existence_strong} implies existence. One-person rationality follows by the fact that the drifting preference relation (used in the definition of strongly maximal states) is {\em rational} in one-player games (that is, reflexive, transitive, and complete). To show consistency we follow a graph-theoretic argument the main idea of which is that since ordinal potential games are ordinally acyclic, they remain so if players outside of any subcoalition are restricted to a strongly equilibrium strategy.

More precisely, let $\Gamma \in O(\mathcal{G})$, let $J$ be a proper subcoalition of the player set of $\Gamma$, and let $s^* \in \phi(\Gamma)$. By Theorem~\ref{opg}, $\Gamma$'s ordinal deployment graph contains no positive cycles. Furthermore, the ordinal deployment graph of $\Gamma^{J, s^*}$ is a subgraph of $\Gamma$'s ordinal deployment graph. Therefore, $\Gamma^{J, s^*}$'s ordinal deployment graph contains no positive cycles either, which implies, also by Theorem~\ref{opg}, that $\Gamma^{J, s^*}$ is an ordinal potential game (that is, $\Gamma^{J, s^*} \in O(\mathcal{G})$). 

It remains to show that if $s^*$ is a strongly maximal equilibrium of $\Gamma$, then, for any proper subcoalition $J$ of players, $s^*_J$ is a strongly maximal equilibrium of $\Gamma^{J, s^*}$. Since $s^*$ is maximal in $\Gamma$ there is no advancement path in $\Gamma$'s ordinal deployment graph whose tail is $s^*$. But since the advancement graph of $\Gamma^{J, s^*}$ is a subgraph (using the standard definition as a subset of vertices and their incident arcs) of $\Gamma$'s advancement graph, the set of paths in $\Gamma^{J, s^*}$ is a subset of the set of paths in $\Gamma$. Therefore, there is no advancement path in $\Gamma^{J, s^*}$'s ordinal deployment graph whose tail is $s^*$ either, and, thus, Lemma \ref{help_lemma} implies $s^*_J$ is a strongly maximal equilibrium of $\Gamma^{J, s^*}$.
\end{proof}

\subsection{Weak acyclicity and weak ordinal acyclicity}

The definition of ``weak acyclicity,'' formalized in a study on the evolution of (social) conventions \cite{Peyton2}, precedes the definition of the sink equilibrium, however, it can also be couched in sink-equilibrium (and, therefore, also maximality) terms: A game in strategic form is called weakly acyclic if the only recurrent states of the Markov chain whose transition graph has arcs corresponding to strictly profitable unilateral deviations are pure Nash equilibria \cite{Fabrikant2}. Rephrasing this definition, a game is called weakly acyclic if weakly maximal states are pure Nash equilibria. 

Weak acyclic (like acyclicity, a notion inherently related to the aforementioned FIP) is closely related to the notion of {\em dynamics} in strategic-form games. In a weakly acyclic game, {\em better response dynamics,} that is, dynamics wherein players respond to the current state by unilaterally taking, in randomly selected turns, better responses, are ensured to converge to a weakly maximal equilibrium state (a pure Nash equilibrium). The importance of weak acyclicity is further argued by related, in this vein, results on dynamics in strategic-form games, namely, that in weakly acyclic games {\em no-regret dynamics} are also guaranteed to converge to a pure Nash equilibrium \cite{Peyton2, MYAS}.

In a manner analogous to the definition of weak acyclicity, we have the following definition.

\begin{definition}
A game in strategic form is called {\em weakly ordinally acyclic} if its strongly maximal states are necessarily strongly maximal equilibria.
\end{definition}

We use this definition in the sequel to analyze the aforementioned election mechanism. Before proceeding with the analysis, we provide a result related to the weak acyclicity of the stag hunt.

\section{Stag hunts are weakly acyclic games}
\label{stag_hunt_is_weakly_acyclic}

The stag hunt has three Nash equilibria, the aforementioned pure equilibria corresponding to universal defection and universal adoption as well as one Nash equilibrium in mixed strategies, which (as mentioned earlier) is generally ignored in the literature. To a large extent this adds support to our aforementioned argument that mixed strategies are not well-suited to serve as an analytical framework for the class of games that are of interest in this paper, which motivated our focus on alternative analytical techniques based on pure-strategy dynamics that formalize the empirical notion of incremental deployability capturing the essence of how the term is used by the networking community. It is natural then to ask if the stag hunt has maximal states other than the Nash equilibria of universal adoption and universal defection. In this section, we show that it doesn't. Figure \ref{saldkjfnalxdkhjff} demonstrates by means of an example of a stag hunt wherein a transient cycle appears, that the stag hunt does not have the FIP, implying it is not a generalized ordinal potential game. However, in the following theorem, we show that cycles in stag hunts are necessarily transient.

\begin{figure}[tb]
\centering
\includegraphics[width=5.5cm]{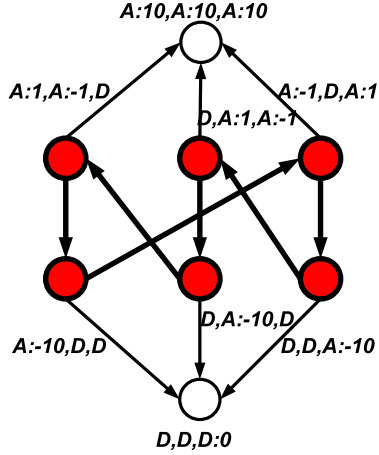}
\caption{\label{saldkjfnalxdkhjff}
Example of a stag hunt with a cycle. The figure shows the deployment graph of the game as well as the payoffs corresponding to each profile. There are three players. The payoff to a player that uses $D$ is $0$ irrespective of what other players do. The payoff of playing $A$ is shown in the figure for each profile in the game.}
\end{figure}

\begin{theorem}
\label{asldkfklfkjjasdp}
Stag hunts are weakly acyclic games.
\end{theorem}

\begin{proof}
We show that starting from any state we are ensured, by following profitable unilateral devisions, to reach either $A^n$ or $D^n$, where $n$ is the number of players. So let us start from an arbitrary state other than $A^n$ and $D^n$ and let us iteratively switch to $D$ all $A$-players that prefer $D$ to $A$. This iterative process may lead to $D^n$. If it does not, then a state is reached in which all $A$-players prefer $A$. Since the only pure Nash equilibria are $A^n$ and $D^n$, a $D$-player must exist that prefers $A$. Let us switch such a player to $A$. Then, by the definition of the stag hunt, all $A$-players in this new state prefer $A$ to $D$. Unless this new state is $A^n$, there must again exist an $D$-player that prefers $A$. Continuing in this fashion, the state $A^n$ is ensured to be reached.
\end{proof}

We note that, using our aforementioned definition of the stag hunt, the distinction between weakly and strongly maximal states vanishes as all unilateral deviations are, by definition, either (strictly) profitable or harmful, and, therefore, the stag hunt is also weakly ordinally acyclic. We believe it is also worth considering variants of the stag hunt where the aforementioned assumption is lifted (in the sense that if a player switches from the noncooperative to the cooperative strategy not all players having adopted the cooperative strategy (strictly) benefit), but we leave the corresponding analysis of such (plausible in some settings) variants as future work.

\section{Analysis of the insurance and election mechanisms}
\label{Analysis}

In this section, we analyze the insurance and election mechanisms from the perspective of Nash equilibrium, dominance solvability, and maximality. In our analysis, we find that the insurance mechanism provides somewhat stronger analytical guarantees, however, in practice it may be somewhat less attractive than the election mechanism (largely due to the requirement of an insurance carrier committing, if not spending, a potentially significant budget to insure deployment efforts).

\subsection{Analysis of the insurance mechanism}

The following analysis of the insurance mechanism is based on the assumption that corresponding contracts (between the carrier and potential adopters) are designed such that, in the event of a coordination failure, the insurance carrier offers a reimbursement amount to an insured player that slightly exceeds that player's deployment cost. Observe that this mechanism induces a game where each player has three strategies (instead of two), namely, to adopt without purchasing insurance (strategy $A$), to defect without purchasing insurance (strategy $D$), and to purchase insurance (strategy $X$). Recall that once a player purchases insurance she is bound by the respective contract to bear the deployment cost. Under the previous assumptions, we have the following theorem.

\begin{theorem}
\label{sid_result}
The insurance mechanism induces a game solvable by strict iterated dominance whose solution is universal adoption without purchasing insurance.
\end{theorem}

\begin{proof}
Observe that strategy $X$ strictly dominates strategy $D$ as it yields a higher payoff irrespective of the strategies of other players (by the assumption that reimbursement under coordination failure slightly exceeds the deployment cost). Once $D$ is eliminated, strategy $A$ strictly dominates strategy $X$, as strategy $X$ yields a payoff lower than that corresponding to universal adoption due to the insurance premium from the players to the insurance carrier. The theorem, therefore, follows.
\end{proof}

Observe in the previous proof that universal adoption is not a dominant strategy but rather two rounds of elimination are required to obtain universal adoption as the solution. Let us now analyze the insurance mechanism from the perspective of maximality. We have the following theorem.

\begin{theorem}
The insurance mechanism induces a game wherein the only (strongly and weakly) maximal state is universal adoption without purchasing insurance.
\end{theorem}

\begin{proof}
As mentioned earlier, in a game that is solvable by strict iterated dominance, the (necessarily) unique outcome of iterated elimination is the unique Nash equilibrium of the game. Therefore, Theorem \ref{sid_result} implies, $A^n$, where $n$ is the number of players, is the unique Nash equilibrium of the induced game. Therefore, universal adoption is a weakly maximal state. Furthermore, universal adoption is also strongly maximal as any unilateral deviation from $A^n$ is harmful. (Unilaterally deviating from $A$ to $X$ is harmful due to the insurance premium whereas unilaterally deviating from $A$ to $D$ is harmful by the definition of the stag hunt.) It remains to show that the induced game is weakly acyclic. To that end, we show that starting from any state, say $s$, $A^n$ is incrementally deployable against $s$. But this is a straightforward implication of the fact that switching from $D$ to $X$ is (strictly) profitable for the corresponding player, and, once all players have switched to $X$, it is also (strictly) profitable for each player to switch from $X$ to $A$. This completes the proof.
\end{proof}

\subsection{Analysis of the election mechanism}

In this section, we analytically demonstrate the election mechanism's efficacy.

\subsubsection{Preliminaries}

The election mechanism induces a game with four strategies: to adopt without voting (strategy $A$), to defect without voting (strategy $D$), to vote and adopt if the outcome of the vote is positive defecting otherwise (strategy $X$), and to vote and adopt irrespective of the outcome of the vote (strategy $Y$). In fact, there are two additional possible strategies that involve the possibility of breaking the commitment that, by design of the mechanism, a positive vote implies; assuming a large enough penalty in the event of breaking the commitment, we may safely ignore these strategies.

\subsubsection{Dominance solvability}

From the perspective of dominance solvability, we have the following theorem.

\begin{theorem}
The election mechanism induces a game that is dominance solvable by weak iterated dominance whose solution corresponds to the strategy profile wherein all players adopt a combination of strategies $X$, that is, to vote and adopt if the outcome of the vote is positive defecting otherwise, and $Y$, that is, to vote and adopt irrespective of the outcome of the vote. 
\end{theorem}

\begin{proof}
Note first that strategy $X$ weakly dominates strategy $D$: Since voting and defecting if the outcome of the vote is negative, switching from $D$ to $X$ implies that the respective player cannot lose, but also (strictly) benefits if all players choose $X$. We may, therefore, eliminate strategy $D$. The remaining strategies are $A$, $X$, and $Y$. Note now that strategy $Y$ weakly dominates strategy $A$: Selecting strategy $Y$ provides the same payoff to the respective player as $A$ except in the event that some other player has chosen strategy $X$, in which case strategy $Y$ provides a (strictly) higher payoff that $A$ since the outcome of the vote becomes negative implying the players having chosen $X$ defect, therefore, also implying a lower payoff to the adopters by the definition of the stag hunt. We may, therefore, also eliminate $A$. The strategies that remain are $X$ and $Y$.
\end{proof}

Note that if all players adopt a combination of $X$ and $Y$, the outcome implies universal adoption. Note further that by the narrow definition of dominance solvability (that a unique strategy remains for each player) the election mechanism does not qualify to be called as such, however, the previous theorem predicts manifestation of the desirable outcome in the basis stag hunt game.

\subsubsection{Maximality}

Let us start the analysis of the election mechanism with respect to maximality with an example in a stag hunt with two players. The incentive structure of the induced game as well as the strongly maximal solutions are illustrated in Figure \ref{exem}. As shown in the figure, all strongly maximal states entail adoption of the cooperative strategy, however, universal defection remains a (weak) pure Nash equilibrium. This point is discussed further at the end of this section. In general stag hunt games, the efficacy of the election mechanism is argued by the following theorem.

\begin{figure}[tb]
\centering
\includegraphics[width=9cm]{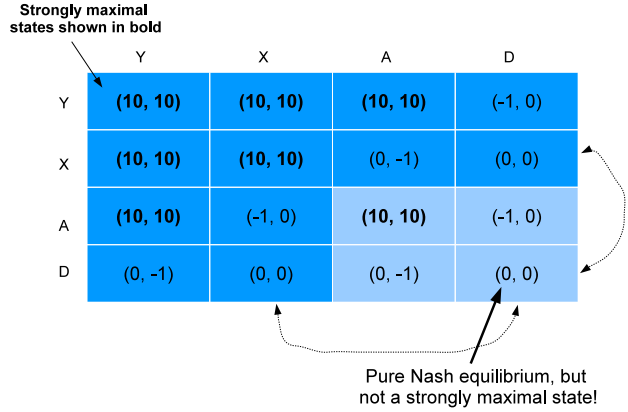}
\caption{\label{exem}
Example of incentive structure the election mechanism induces in a (two player) stag hunt.}
\end{figure}

\begin{theorem}
The election mechanism induces a game that is weakly ordinally acyclic all of whose strongly maximal equilibria imply manifestation of the superior outcome of the basis stag hunt game.
\end{theorem}

\begin{proof}
Consider the ordinal deployment graph and a state, say $s$, such that one or more players select strategy $D$. Observe that $D$-players never lose by switching from $D$ to $X$, and, therefore, switching all $D$-players in this way implies an outcome, say $s'$, that is incrementally deployable against $s$. Let us now focus attention to the players in $s'$ that select $A$. Since such players do not vote, players that select $X$ defect. However, switching $A$-players to $X$ implies that they never lose and in fact benefit once all players have been switched from $A$ to $X$. Call this latter state $s''$, and note that $s''$ is incrementally deployable against $s'$ whereas $s'$ is not so against $s''$. Since the ordinal preference relation is transitive, this implies that $s''$ is incrementally deployable against $s$ and, it is easy to observe, that $s$ is not so against $s''$. Therefore, strategy $D$ never manifests in any state that is (strongly) maximal. Note finally that $A^n$ is also a strongly maximal state: Switching players from $X^n$ to $Y^n$ implies no loss in payoff, and, therefore, that $Y^n$ is incrementally deployable against $X^n$. Furthermore, switching players from $Y^n$ to $A^n$ also implies no loss in payoff, and therefore that $A^n$ is incrementally deployable against $Y^n$ and by transitive also against $X^n$. Thus, $A^n$, combinations of $A$ and $Y$, and combinations of $X$ and $Y$ are strongly maximal equilibria belonging to the same equilibrium class. Finally observe that as noted above no state that includes $D$ is maximal.
\end{proof}

\subsubsection{Some remarks}

We note that universal defection remains a pure Nash equilibrium in the game that the election mechanism induces and, therefore, universal defection is also a weakly maximal state. We also note, however, that universal defection is a {\em weak} Nash equilibrium after the election mechanism is introduced in contrast to being a {\em strict} Nash equilibrium in the basis game (the stag hunt on which the election mechanism is applied), and naturally the predictive value of weak Nash equilibria is (in general) weaker than those of their strict counterparts. In the particular setting we consider, the ``evolutionary path'' from universal defection to universal adoption is one where no player loses during the transition while all players benefit after the transition has taken place. 

We finally note that although weak iterated dominance predicts that all players will vote in favor of adoption (implying that the mechanism will commit them to do so), the analysis based on strong maximality is inconclusive on whether players will vote; the outcome of universal adoption without participating in the voting procedure is as plausible an outcome as being committed to adoption by a universal vote. However, both solution concepts (namely, weak iterated dominance and strong maximality) agree on the final outcome that universal adoption will manifest (whether by means of the commitment mechanism or not). Empirical observation of the outcome in experimental settings is, therefore, of natural interest, and an interesting direction for future work.

\section{Discussion}
\label{discussion}

We started out this paper motivating our contributions using examples of innovation failures in the Internet architecture such as the lack of the deployment of countermeasures against routing attacks in BGP and the extremely slow pace of replacing IPv4 with IPv6. In both cases, the success of the emerging technologies (BGPsec and IPv6, respectively) depends on positive externalities in the sense that early adopters do not benefit. Then we attributed these failures to an equilibrium selection phenomenon using a game-theoretic formulation, in particular, the stag hunt, and we presented techniques aiming to facilitate selection of the superior equilibrium in such games. We should note, however, that equilibrium changes do sometimes happen in practice using ``evolutionary forces'' alone (without any exogenous mechanism being in place to facilitate such a transition). 

For example, Facebook, a social network whose success depends on positive externalities, succeeded in attracting users from its predecessors such as MySpace, although Facebook's early adopters did not receive any significant benefit from signing up and interacting with a sparse community of early users. These transitions admit game-theoretic explanations, using, for example, equilibrium selection theories such as that introduced by Harsanyi and Selten \cite{HS}. According to these authors, which equilibrium is selected in a stag hunt largely depends on two factors, namely, {\em risk} and {\em benefit} that both weigh in in predicting the outcome players choose. (Note that in a stag hunt there is no risk in defecting, as defection guarantees a constant payoff irrespective of what other players do, however, playing the cooperative strategy entails the possibility of a loss if sufficiently many other players do not cooperate.) If the benefit is too large compared to the risk, then an equilibrium transition from the inferior to the superior equilibrium becomes possible.

In this vein, the efforts of networking researchers to address the Internet's innovation slump by devising increasingly better technologies are justifiable in the sense that as the performance and security, for example, characteristics of emerging technologies improve, their adoption becomes more likely using evolutionary forces alone. However, technological advances are hard to project in the future, as is the probability of these designs being adopted in practice (owing to a lack of precise predictions by equilibrium selection theories). Coordination mechanisms, in contrast, can be designed to facilitate such transitions without significantly, if at all, depending on the precise characteristics of the risk and payoff structure of adoption environments.

This observation is akin to a phenomenon that manifests in road traffic transportation networks. In such networks, engineers have observed that the input traffic, for example, in the road traffic transportation of major cities, has the desirable property of adjusting (without exogenous interventions) to congestion levels. Nevertheless, this adaptation can be slow, and, therefore, exogenous interventions, such as {\em tolls,} are used as a means of controlling congestion extensively around the world. Tolls bear an analogous functionality as the coordination mechanisms that we propose in this paper: The model in widespread use by researchers in civil engineering and operations research to predict traffic patterns in transportation systems is one proposed by Wardrop whose predictions are known as {\em Wardrop equilibria}. This model admits a game-theoretic formalization where Wardrop equilibria correspond to Nash equilibria in a class of games known as {\em non-atomic congestion games}. Noting that this model was popularized in the computer science literature by the work of Roughgarden and Tardos \cite{HowBadisSelfishRouting} on the price of anarchy \cite{PriceOfAnarchy}, tolls serve a similar function to that coordination mechanisms aim to achieve in this paper, namely, they serve as a means of selecting a more desirable equilibrium than the one evolutionary forces alone predict. In this sense, the notion of tolls qualifies as a game-theoretic mechanism, albeit, in principle, more intrusive than the mechanisms we present in this paper, which are by nature based on a ``opt-in'' philosophy.

We believe that auctions, tolls, and coordination mechanisms are design artifacts that are in a sense {\em dual} to technological artifacts as they are in a position to complement or even replace such technological counterparts toward the common objective of engineering better social outcomes. For example, coming back to our previous example on road traffic transportation networks, instead of establishing a charging system based on tolls, transportations authorities may alternative decide to increase the carrying capacity, at a potentially significantly higher cost, however. Naturally, analogous design tradeoffs can manifest in a variety of social planning environments.

\section{Other related work}
\label{other_rel_work}

The problem of effecting or diffusing innovation in social systems has been the subject of extensive scrutiny primarily by social scientists but also more recently by computer scientists. The term {\em innovation} assumes various possible manifestations in the literature from agricultural practices, social norms (such as which hand to extent in a handshake), medical drugs, commercial products (such as fax machines and cellphones), to networking technologies such as secure versions of BGP (like BGPsec) and quality-of-service capabilities in IP networks. In a variety of diffusion models, potential adopters are assumed to interact according to network structures that are broadly referred to as social networks (whether networks based on kinship or friendship or online social networks enabled by the widespread penetration of the Internet). An excellent reference tracking the history of research in this area but also concerned with recent contributions is by Rogers \cite{Diffusion}.

Various researchers, especially computer scientists, have recently been concerned with an algorithmic problem related to the diffusion of innovation typically referred to as {\em influence maximization} initially studied by Domingo and Richardson \cite{Domingos, Richardson}. Influence maximization refers to selecting an initial set of adopters in an adoption environment characterized by network structure such that the corresponding innovation will spread, by forces of evolution characterizing the respective environment, to an as large fraction of the population as possible (possibly the entire population). (The diffusion model typically assumed in this line of work is based on a notion of ``value'' the innovative technology or behavior yields that, for each potential adopter, is dependent on the state of other potential adopters.) Kempe {\em et al.} \cite{Spread} as well as Goldberg and Liu \cite{Goldberg2} pose this problem as a combinatorial optimization problem, shown to be NP-hard, and devise approximation algorithms in obtaining efficient solutions (that have more recently been improved in followup studies). This line of research leaves open, however, the issue of how initial adopters are enticed to adopt.

In contrast to these earlier works, our model of evolution based on the stag hunt, although it was motivated by the diffusion of (Internet-based) networking technologies that have a natural network structure, does not make specific assumptions on the structure of the adoption environment other than adopters benefitting from the adoption decisions of others. In spite of its generality, we show that coordination mechanisms incite widespread diffusion while at the same time escaping the fundamental question (central in previous work) of kickstarting the diffusion process, which we achieve by the mere act of instituting coordination mechanisms in adoption environments.

Related to the problem of innovation diffusion is that of inciting {\em collective action,} a problem whose study was initiated in the seminal work of Olson \cite{Olson}. The study of collective action was later taken up by several authors such as, for example, in the form of {\em critical mass theories} (bearing relevance to nuclear fission explosions in physics) \cite{OMT, Markus}. Medina \cite{Medina} motivates the study of collective action in settings bearing a political nature related to citizen oppression in rogue states, however, models of collective action range from crowdsourcing in the Internet to enabling countermeasures against climate change (such as controlling carbon emissions). Although both Olson and Medina discuss the problem of inciting collective action, the idea of inducing (institutional) mechanisms to that effect (as we do in this paper) admitting a rigorous game-theoretic analysis has not, to the extent of our knowledge, been considered before by these earlier works.

\section{Conclusion and future work}
\label{conclusion}

Motivated by the innovation slump in the Internet's basic architecture, in this paper we presented a notion of game-theoretic mechanisms, unlike those used and analyzed in mechanism design theory (whose main goal is to entice a truthful revelation of unknown to the designer preferences), aiming instead to incite equilibrium transitions in environments whose incentive structure can be modeled after the stag hunt, a well-known coordination game that has found applications in various social and economic settings. We proposed two specific instances of such coordination mechanisms, an insurance and an election mechanism, that we thoroughly analyzed using a variety of game-theoretic solution concepts. Our analytical techniques were motivated and developed by colloquial notions of deployability in the Internet, which we rigorously developed and placed squarely within the broader game-theoretic analytical apparatus. We believe that a fruitful direction for further research is to empirically study and validate the efficacy of these mechanisms. Many practical questions remain to be analyzed in this vein such as, for example, the extent to which potential adopters being given the opportunity to purchase an insurance policy from a carrier would purchase insurance policies from that carrier and whether adopters may exist that would prefer to adopt by eschewing the purchase of insurance as the strict iterated dominance analysis predicts. Naturally, analogous questions can be raised in the presence of an election mechanism; would potential adopters prefer to vote in this event or rather adopt eschewing the possible commitment that voting may imply?

\section*{Acknowledgments}

I would like to thank various members of the Erato project at the National Institute of Informatics (NII) in Tokyo, Japan for attending a presentation related to the topic of this paper and offering helpful suggestions. Earlier versions of this paper have benefitted from discussions with Jon Crowcroft, Miltos Anagnostou, and Jen Rexford. I would, finally, like to thank an anonymous reviewer for simplifying my proof that stag hunts are weakly acyclic games.

\bibliographystyle{plain}
\bibliography{coord}

\end{document}